\documentclass[a4paper,11pt]{article}
\usepackage{amssymb,amsmath,latexsym,amsthm}
\usepackage{color,mathrsfs}
\usepackage{url}
\usepackage{enumerate}
\usepackage[square,sort,comma,numbers]{natbib}
\usepackage{dsfont}
\usepackage{diagbox}
\usepackage{graphicx}
\usepackage{colortbl}

\usepackage{endnotes}

\renewcommand{\textbf}[1]{\begingroup\bfseries\mathversion{bold}#1\endgroup}

\newlength{\bibitemsep}\newlength{\bibparskip}
\let\oldthebibliography\thebibliography \renewcommand\thebibliography[1]{
  \oldthebibliography{#1} \setlength{\parskip}{\bibitemsep}
  \setlength{\itemsep}{\bibparskip} }

\usepackage{fancyhdr}
\usepackage{pstricks,pst-plot,pst-node,pstricks-add}

\usepackage[left=1.95cm, right=1.95cm, top=2cm]{geometry}

\newtheorem{thm}{Theorem}[section]
\newtheorem{defi}{Definition}[section]

\newtheorem{prop}[thm]{Proposition}
\newtheorem{Conjecture}[thm]{Conjecture}
\newtheorem{lemma}[thm]{Lemma}

\theoremstyle{definition}
\newtheorem{remark}[thm]{Remark}

\newtheorem{examples}[thm]{Examples}

\newcommand{\R}{\mathbb R}

\newcommand{\Z}{\mathbb Z}
\newcommand{\N}{\mathbb N}

\numberwithin{equation}{section}

\def\XXint#1#2#3{{\setbox0=\hbox{$#1{#2#3}{\int}$}
    \vcenter{\hbox{$#2#3$}}\kern-.5\wd0}}

\allowdisplaybreaks
\date{date}
\begin{document}
\title{Optimality of the triangular lattice for Lennard-Jones type lattice energies: a computer-assisted method}
\author{Laurent B\'{e}termin\\ \\
Institut Camille Jordan, Universit\'e Claude Bernard Lyon 1\\ 43 Bd du 11 novembre 1918, 69622 Villeurbanne Cedex, France\\ \texttt{betermin@math.univ-lyon1.fr}. ORCID id: 0000-0003-4070-3344 }
\date\today
\maketitle

\begin{abstract}
It is well-known that any Lennard-Jones type potential energy must have a periodic ground state given by a triangular lattice in dimension 2. In this paper, we describe a computer-assisted method that rigorously shows such global minimality result among $2$-dimensional lattices once the exponents of the potential have been fixed. The method is applied to the widely used classical $(12,6)$ Lennard-Jones potential, which is the main result of this work. Furthermore, a new bound on the inverse density (i.e. the co-volume) for which the triangular lattice is minimal is derived, improving those found in [L. B\'etermin and P. Zhang, \textit{Commun. Contemp. Math.}, 17 (2015), 1450049] and [L. B\'etermin, \textit{SIAM J. Math. Anal.}, 48 (2016), 3236--3269]. The same results are also shown to hold for other exponents as additional examples and a new conjecture implying the global optimality of a triangular lattice for any parameters is stated.
\end{abstract}



\section{Introduction and main results}

The question of existence and uniqueness of periodic ground states for interacting particles systems has been addressed in many contexts (see \cite{BlancLewin-2015}). The crucial question ``Why are solids crystalline?" asked for instance by Radin \cite{RadinLowT} in the context of low-temperature matter is still under investigation and only few rigorous answers have emerged, leaving the ``Crystallization Conjecture" largely open. Even when particles are interacting through a pairwise potential, finding theoretically or numerically the optimal configuration of points minimizing the potential energy of such system is extremely challenging due to the huge number of parameters and the presence of lots of critical points.

\medskip

We usually expect for a pairwise interaction potential between radially symmetric neutral particles at very low temperature to be repulsive at small distances (due to Pauli Principle) and attractive at large distances (due to Van der Waals forces). A prototypical example of such interaction is the classical Lennard-Jones potential
\begin{equation}\label{eq:classicLJ}
f(r)=\frac{a}{r^{12}}-\frac{b}{r^6}, \quad \textnormal{where}\quad (a,b)\in (0,\infty)^2,
\end{equation}
popularized by Lennard-Jones in \cite{LJ} to study the thermodynamical properties
of rare gases (e.g. liquid argon), but originally proposed by Mie in \cite{Mie} in the general form $f(r)=ar^{-\alpha}-br^{-\beta}$ for $\alpha>\beta$ (we will however call such potential ``Lennard-Jones type potentials"). The repulsion is therefore given by the inverse power law $a r^{-12}$ and the behavior at infinity in $-b r^{-6}$ mimics the Van der Waals attraction. This potential is widely used in Molecular Dynamics (see e.g. \cite{Kaplan,MolecSimul}). In dimension 2, particles interacting through such potential are expected to crystallize, i.e. to find their ground state, on a perfect triangular lattice (see Figure \ref{fig:trilattice} and \eqref{eq:tri}) as their number goes to infinity (see e.g. \cite{YBLB} and \cite[Fig. 1]{BlancLewin-2015}). In dimension 3, the same is expected to happen on a Hexagonal Close Packing structure (see e.g. \cite{ModifMorse,Beterminlocal3d}). The rare attempts to show the two-dimensional case have lead to optimality results for the triangular lattice in the case of hard-sphere potentials \cite{Rad2,Luca:2016aa} and approximations of them \cite{Rad3,Crystal,AuyeungFrieseckeSchmidt-2012}. Notice also that crystallization results for long-range pairwise potentials have only been proven in \cite{Crystal} (on a triangular lattice) and recently in \cite{BDLPSquare} (on a square lattice) where the potentials are strongly repulsive at the origin and have very narrow wells (for the triangular lattice, in order to catch only the nearest-neighbors) or sufficiently large well one (for the square lattice, in order to also catch the next-nearest neighbors) before converging rapidly to zero.

\begin{figure}[!h]
\centering
\includegraphics[width=5cm]{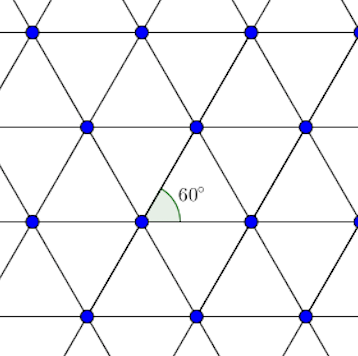}
\caption{Patch of the triangular lattice of unit density $\mathsf{A}_2\subset \R^2$ defined by \eqref{eq:tri} where all primitive triangles are equilateral with side length $\sqrt{2/\sqrt{3}}$.}
\label{fig:trilattice}
\end{figure}

\medskip

When crystallization is assumed, finding the most energetically favorable periodic lattice structure (see Definition \ref{def:lattices}) is an interesting and much simpler problem due to the small number of parameters in action. However, the apparent simplicity of the problem hides a lot of technical issues, even in dimension 2. Indeed, the usual way to minimize energies of type 
\begin{equation}\label{def:EfL2}
E_f[L]:=\sideset{}{'}\sum_{p\in L} f(|p|),\quad L\in \mathcal{L}_2:=\{\Z u \oplus \Z v\subset \R^2 : (u,v) \textnormal{ is a basis of $\R^2$}\},
\end{equation}
where $|\cdot|$ is the Euclidean norm on $\R^2$ and $\sideset{}{'}\sum$ denotes the sum without the $p=0$ term, among lattices $L\in \mathcal{L}_2$ is to parametrize the lattices with 3 real numbers and to optimize this 3-variable function by computing its derivatives. When $f(r)=r^{-s}$ is an inverse power-law, $E_f$ is called the Epstein zeta function defined, for $s>2$ and $L\in \mathcal{L}_2$, by
\begin{equation}\label{intro:zeta}
\zeta_L(s):=\sideset{}{'}\sum_{p\in L}\frac{1}{|p|^s},
\end{equation}
 and has been studied by Rankin \cite{Rankin}, Ennola \cite{Eno2}, Cassels \cite{Cassels} and Diananda \cite{Diananda} (see also \cite{Henn} for a review of these results). They have shown, for all fixed $s>2$, the minimality of the triangular lattice  $\mathsf{A}_2\subset \R^2$ defined by
\begin{equation}\label{eq:tri}
\mathsf{A}_2:=\sqrt{\frac{2}{\sqrt{3}}}\left[\Z (1,0)\oplus \Z\left( \frac{1}{2},\frac{\sqrt{3}}{2}\right)  \right]
\end{equation}
among two-dimensional lattices of fixed unit density. This result have been generalized to all $f(r)=F(r^2)$, where $F$ is a completely monotone function (i.e. the Laplace transform of a nonnegative Radon measure), by Montgomery in \cite{Mont} where he showed the optimality of the triangular lattice for any Gaussian function $f(r)=e^{- \alpha r^2}$, $\alpha>0$ (this is also called ``universal optimality among lattices"). Concerning the Lennard-Jones type potentials $f(r)=a r^{-\alpha}- b r^{-\beta}$, several results have been derived in \cite{Betermin:2014fy,BetTheta15,Beterloc,OptinonCM,LBbonds21}, especially the optimality of the triangular lattice among lattices with fixed high density as well as the minimality of a triangular lattice for $E_f$, with small exponents, among all possible two-dimensional lattices. We have also shown in \cite{OptinonCM} that the type (or ``shape") of the global minimizer of $E_f$ is independent of $(a,b) \in (0,\infty)^2$ (see also Theorem \ref{thm:LJ}). Unfortunately, the method of \cite{Betermin:2014fy,BetTheta15} based on Montgomery's result and a Riemann's splitting of $E_f$ was not adapted to show that a triangular lattice is the minimizer of $E_f$ among all lattices when $f$ is the classical Lennard-Jones potential \eqref{eq:classicLJ}, but only for certain pairs of small exponents where $\alpha\leq 8$ (see \cite[Thm. 1.2.B.2 and Rmk. 6.18]{BetTheta15} for more details).

\medskip

The present work aims to correct this issue by deriving a rigorous computer-assisted method designed to show that, given the exponents $\alpha>\beta>2$, the minimizer of $E_f$ is a triangular lattice. This general optimality result has been conjectured in  \cite{Betermin:2014fy,BetTheta15,Beterloc,OptinonCM} and numerically checked several times (see also \cite{SamajTravenecLJ}). Even though we have checked our proof for many different pairs of exponents (see Section \ref{sec:otherexponents} and in particular Equation \eqref{eq:exp}), we state our main theorem for the classical Lennard-Jones potential \eqref{eq:classicLJ} since this is the most popular  and interesting one among them in mathematical physics and molecular dynamics. 

\begin{thm}[\textbf{Optimality of a triangular lattice for the classical Lennard-Jones energy}]\label{mainthm}
For any $(a,b)\in (0,\infty)^2$, the triangular lattice $\sqrt{V_{\mathsf{A_2}}}\mathsf{A}_2$ of co-volume $V_{\mathsf{A_2}}=\left( \frac{2 a  \zeta_{\mathsf{A_2}}(12) }{b  \zeta_{\mathsf{A_2}}(6)}\right)^{\frac{1}{3}}$ is the unique minimizer in $\mathcal{L}_2$ defined by \eqref{def:EfL2}, up to rotation, of 
\begin{equation}\label{def:Ef6-12}
L\mapsto E_f[L]:= \sideset{}{'} \sum_{p\in  L}\left[ \frac{a}{|p|^{12}}- \frac{b}{|p|^6} \right]=a\zeta_L(12)-b\zeta_L(6),
\end{equation}
and where the Epstein zeta function $\zeta_L(s)$ is defined by \eqref{intro:zeta}.
\end{thm}

The proof's strategy is precisely explained in Section \ref{sec:strategy} but we briefly sketch it here. This is a computer-assisted proof since it needs a computer to check a finite number of values in order to conclude. The method is also designed to work for a general pair of exponents $(\alpha,\beta)$ replacing the parameters $(12,6)$ in Theorem \ref{mainthm}, i.e. when $f$ is a Lennard-Jones type potential $f(r)=a r^{-\alpha}- b r^{-\beta}$. First, our minimization problem among all lattices is transformed into the new problem $\mathcal{Q}_{\alpha,\beta}(L)>\alpha/\beta$, where, for all $L\neq \mathsf{A}_2$,
$$
\mathcal{Q}_{\alpha,\beta}(L):=\frac{\zeta_L(\alpha)-\zeta_{\mathsf{A}_2}(\alpha)}{\zeta_L(\beta)-\zeta_{\mathsf{A}_2}(\beta)},
$$
among unit density lattices. We therefore restrict our problem to a compact set $\overline{\mathcal{K}_{\alpha,\beta}}$ (see \eqref{eq:Kbar} and Figure \ref{fig:ex}) of unit density lattices, since the above inequality appears to be true outside of it. This is hence possible to design a square grid of points $G_{\alpha,\beta}^\delta\subset\overline{\mathcal{K}_{\alpha,\beta}}$ (see Figure \ref{fig:ex}) with appropriate small size $\delta$ such that the minimum value of $\mathcal{Q}_{\alpha,\beta}$ in $G_{\alpha,\beta}^\delta$ gives a sufficiently precise value of the minimum of $\mathcal{Q}_{\alpha,\beta}$ in $\overline{\mathcal{K}_{\alpha,\beta}}$. Finally, if $\mathcal{Q}_{\alpha,\beta}(L)>\frac{\alpha}{\beta}$ among these values with the appropriate degree of accuracy, then Theorem \ref{mainthm} is proved for the pair of exponents $(\alpha,\beta)$. This is in particular true for $(\alpha,\beta)=(12,6)$ as precisely stated in Theorem \ref{mainthm} (see Section \ref{sec:126}) and for other exponents as shown in Section \ref{sec:otherexponents}.

\medskip

This type of computer-assisted proof applied to lattice energy minimization problems has been also used by Sarnak and Strombergsson in \cite{SarStromb} for showing that the FCC lattice is the minimizer of the three-dimensional height among flat tori (i.e. basically the derivative of the Epstein zeta function at $s=0$). However, our problem concerning $\mathcal{Q}_{\alpha,\beta}$ is really less complicated than the latter since we are in dimension 2 (i.e. the computational time is rather short) and we do not need extremely precise estimates to conclude (the inequality  $\mathcal{Q}_{\alpha,\beta}(L)>\frac{\alpha}{\beta}$ is actually far from being sharp, see Sections \ref{sec:126} and \ref{sec:otherexponents}). For these reasons, we do not provide any code in this paper since we only need to compute a good approximation of Epstein zeta functions (i.e. with enough terms) as well as enough values on the grid $G_{\alpha,\beta}^\delta$ for showing our result, which are rather simple tasks. Also, our method can be applied to any chosen parameters $(\alpha,\beta)$ (see Section \ref{sec:otherexponents}) as it is expected that the minimality of the triangular lattice holds for all exponents (see e.g. \cite{OptinonCM,SamajTravenecLJ}). However, our proof is specific to Lennard-Jones type potentials since it strongly relies on homogeneity of the Epstein zeta function (see Lemma \ref{lem:scale}) as well as on the fact that $f$ is a one-well potential. Therefore, we do not expect our method to be adaptable to other types of potentials. Furthermore, since no optimality result is available for the three-dimensional Epstein zeta function (for which the Face-Centred-Cubic lattice is expected to be a minimizer when $s>3$) and since the Hexagonal-Close-Packing structure is not a lattice in the sense of our Definition \ref{def:lattices}, our method is not directly applicable in dimension $d=3$ (see also Remark \ref{rmk:nonadapt}).

\medskip

It is interesting to notice that Theorem \ref{mainthm} -- and its analogue for any $(\alpha,\beta)$ where the computer-assisted proof is checked -- directly implies optimality results of a triangular lattice (with the corresponding optimal scaling) for Embedded-Atom Models of the form
$$
\mathcal{E}(L):=F\left( E_\rho[L]\right) + E_f[L],\quad L\in \mathcal{L}_2,\quad \textnormal{$E_f$ and $E_\rho$ defined by \eqref{def:EfL2}}, 
$$
and where, for all $r>0$, $\rho(r)=r^{-\beta}$, $f(r)\in \{ r^{-\alpha}, a r^{-\alpha}-b r^{-\beta} \}$ and $F(r)\in \{r^t,-\log r, -c\sqrt{r} \}$, with $(c,t)\in (0,\infty)^2$ and $\alpha>\beta>2$, as recently shown in \cite[Thm. 4.1 and 5.2]{LBEAM21}. The function $\rho$ (resp. $f$) corresponds to the electrons (resp. nuclei) energy contribution of atoms located at lattice sites. This type of semiempirical model, initially introduced by Daw and Baskes \cite{DawBaskes} and which can be  viewed as an intermediate model between zero-temperature phenomenological pair-interaction energies and quantum systems, is routinely used in Molecular Simulation for Material Sciences (see e.g. \cite{EAMReview,LeSar}).

\medskip

Additionally, another consequence of our proof is the following new bound for the co-volume $V$ (i.e. the inverse density) for which $\sqrt{V}\mathsf{A}_2$ is the unique minimizer of $E_f$ in $\mathcal{L}_2(V)$, i.e. among lattices with fixed co-volume $V$ (see Definition \ref{def:lattices} for more details).

\begin{prop}[\textbf{Optimality at fixed high density of the triangular lattice}]\label{prop:bound}
For any $(a,b)\in (0,\infty)^2$, if $0<V\leq \left( \frac{2 a}{b }\right)^{\frac{1}{3}}$ then the triangular lattice $\sqrt{V}\mathsf{A}_2$ is the unique minimizer of $E_f$ defined by \eqref{def:Ef6-12} in $\mathcal{L}_2(V)$ up to rotation.
\end{prop}

This result improves the previous bound we found in \cite{Betermin:2014fy,BetTheta15}, which was $\left( \frac{ \pi^3 a}{60 b}\right)^{\frac{1}{3}}$. Furthermore, our proof is actually designed to show that the general bound for fixed $(\alpha,\beta)$ and any $(a,b)$ is actually $\left(\frac{a\alpha}{b\beta}  \right)^{\frac{2}{\alpha-\beta}}$ (see Section \ref{sec:strategy} and Section \ref{sec:otherexponents} for applications to other exponents), which again improves the bound of \cite[Prop. 6.11]{BetTheta15} obtained by using a Riemann's splitting argument. Notice that this bound is still not the sharpest one since we expect the triangular lattice to be minimal in $\mathcal{L}_2(V)$ for slightly larger values of $V$ (see  \cite{Beterloc,SamajTravenecLJ}). It is also important to notice that Proposition \ref{prop:bound} actually implies Theorem \ref{mainthm}. However, we have chosen to write the results in this order for highlighting our main global optimality result.

\medskip

We would like to add to this introduction the statement of an important conjecture related to our problem, which implies Theorem \ref{mainthm} and Proposition \ref{prop:bound} for any pair of exponents. Indeed, the present work shows a method to prove the optimality of the triangular lattice for any Lennard-Jones type potentials, but only once the exponents $(\alpha,\beta)$ are fixed. It would be obviously very interesting to find a more general proof since it is expected that these results hold for all exponents (see \cite{BetTheta15,Beterloc,OptinonCM}). Theorem \ref{mainthm} and Proposition \ref{prop:bound} would actually hold for any parameters $2<\beta<\alpha$ if we were able to show the following conjecture.
\begin{Conjecture}\label{conj}
For all $s>2$, $\mathsf{A}_2$ is the unique minimizer, among the two-dimensional lattices of unit density, of
$$
\mathcal{F}_s(L):= -\sideset{}{'} \sum_{p\in L} \frac{s\log |p|+1}{|p|^s}.
$$
\end{Conjecture}
Indeed, this conjecture is equivalent to the fact that $\mathcal{Q}_{\alpha,\beta}(L)>\frac{\alpha}{\beta}$ for all unit density lattice $L$ and all $2<\beta<\alpha$  which implies Proposition \ref{prop:bound} and then Theorem \ref{mainthm} for all the exponents (see Section \ref{sec:otherexponents} for details). We have actually already shown this result in \cite[Thm. 1.2.B.2]{BetTheta15} for $2<s<2\left( \psi^{-1}(\log \pi)-1 \right)\approx 5.25$ (see also \cite[Rmk. 6.18]{BetTheta15} for details) but the conjecture stays totally open for larger $s$. Furthermore, since $r\mapsto -(s\log r +1)r^{-s}$ is a one-well potential with minimum at $r=1$, Conjecture \ref{conj} stays however difficult to solve and is therefore an interesting challenging problem.

\medskip

Finally, it has to be noticed that out study's setting does not cover other important physical cases such as periodic topological defect pattern in the background of the lattice structures or various symmetry-breaking amorphous lattices. For example, dislocations may be periodically distributed over the triangular lattice. Considering
the long tail of the Lennard-Jones potential, the above-mentioned two kinds of cases could not be readily excluded as the energy minima. The reader can for example refer to \cite{OrtneretalLJ} and references therein.

\medskip

\textbf{Plan of the paper.} We first recall in Section \ref{sec:setting} the notions of lattices and energies (see Section \ref{sec:latticeenergy}), as well previous results on Epstein zeta function and Lennard-Jones type energies (see Section \ref{sec:LJrecall}). Section \ref{proof} is devoted to the proof of Theorem \ref{mainthm} and Proposition \ref{prop:bound} where the general method is explained in Section \ref{sec:strategy}, the application to the parameters $(\alpha,\beta)=(12,6)$ is shown in Section \ref{sec:126} and to others exponents in Section \ref{sec:otherexponents}. Finally, a central technical auxiliary lemma is proved in Section \ref{subsec:prooflemmas}.

\section{Definitions and previous results on lattice energies}\label{sec:setting}

\subsection{Lattices and Energies}\label{sec:latticeenergy}

We first define the set of periodic discrete configurations of points we are interested in, as well as their associated quadratic forms.

\begin{defi}[Lattices and associated quadratic form]\label{def:lattices}
We call $\mathcal{L}_2$ the set of two-dimensional lattices of the form $L=\Z u \oplus \Z v$ where $\{u,v\}\subset \R^2$ is a basis of $\R^2$. Furthermore, we write $\mathcal{L}_2(V)\subset \mathcal{L}_2$ the set of lattices $L=\Z u \oplus \Z v$ with co-volume $|\det(u,v)|=V$. 

\medskip

\noindent Moreover, to any lattice $L=\Z u \oplus \Z v \in \mathcal{L}_2$ is associated a quadratic form $Q_L$ defined by
$$
Q_L(m,n):=\left| m u + nv \right|^2, \quad \forall (m,n)\in \Z^2.
$$
\end{defi}

Furthermore, the theory of quadratic forms' reduction (see e.g. \cite{Terras_1988} or \cite{Mont}) allows to restrict our set of lattices to a fundamental domain where each of them appears only once and can be parametrized by a point $(x,y)\in \R^2$.

\begin{prop}[Parametrization and fundamental domain]\label{prop:D}
Any lattice $L\in \mathcal{L}_2(1)$ can be uniquely parametrized by $(x,y)\in \mathcal{D}\subset \R^2$ such that
\begin{equation}\label{def:D}
\mathcal{D}:=\left\{(x,y)\in \R^2 : y>0, x\in [0,1/2], x^2+y^2\geq 1  \right\}.
\end{equation}
The set $\mathcal{D}$, depicted in Figure \ref{fig:ex}, is called the half-fundamental domain, the parametrization can be written
$$
L=\Z u \oplus \Z v,\quad \textnormal{where}\quad u=\left( \frac{1}{\sqrt{y}},0 \right) \quad \textnormal{and}\quad v=\left(\frac{x}{\sqrt{y}},\sqrt{y}  \right),
$$
and the quadratic form of such $L$ is therefore
$$
Q_L(m,n)=\frac{1}{y}(m+xn)^2 + y n^2,\quad \forall (m,n)\in \Z^2.
$$
\end{prop}

\begin{examples}[The triangular and square lattices]
The square lattice $\Z^2$ and the triangular lattice $\mathsf{A}_2$, respectively defined by
$$
\Z^2= \Z(1,0)\oplus \Z(0,1) \quad \textnormal{and}\quad \mathsf{A}_2:=\sqrt{\frac{2}{\sqrt{3}}}\left[\Z (1,0)\oplus \Z\left( \frac{1}{2},\frac{\sqrt{3}}{2}\right)  \right],
$$
belong to $\mathcal{L}_2(1)$ and are respectively parametrized in $\mathcal{D}$ by $(0,1)$ and $\left(\frac{1}{2},\frac{\sqrt{3}}{2} \right)$.
\end{examples}

We now define our main lattice energies, namely the Epstein zeta function and the Lennard-Jones type energy which is actually a linear combination of the first one. The Epstein zeta function was originally introduced by Epstein \cite{Epstein1} and appears in many contexts (see e.g. \cite{Sobolev,Henn,BeterminKnuepfer-preprint}).

\begin{defi}[Epstein zeta function and Lennard-Jones type energy]\label{def:EpstLJ}
For any $s>2$, the Epstein zeta function $\zeta_L(s)$ of $L\in \mathcal{L}_2$ and its $N^{th}$ partial sum $\zeta^N_L(s)$, for some $N\in \N$, are defined by
\begin{equation}\label{def:partialEpstein}
\zeta_L(s):=\sideset{}{'}\sum_{p\in L} \frac{1}{|p|^s}=\sideset{}{'}\sum_{(m,n)\in \Z^2} \frac{1}{Q_L(m,n)^{\frac{s}{2}}}\quad \textnormal{and}\quad  \zeta^N_L(s):=\sideset{}{'}\sum_{|m|\leq N, |n|\leq N \atop (m,n)\in \Z^2} \frac{1}{Q_L(m,n)^{\frac{s}{2}}},
\end{equation}
where $\sideset{}{'}\sum$ is the sum where the origin is omitted. 

\medskip

Furthermore, for any $\alpha>\beta>2$ and $(a,b)\in (0,\infty)^2$, the Lennard-Jones type potential is defined by
$$
f(r):=\frac{a}{r^\alpha}-\frac{b}{r^\beta},
$$
and the corresponding Lennard-Jones type energy of $L\in \mathcal{L}_2$ is given by
$$
E_f[L]:=\sideset{}{'}\sum_{p\in  L} f(|p|)=a\zeta_L(\alpha)-b\zeta_L(\beta)=\sideset{}{'} \sum_{p\in  L}\left[ \frac{a}{|p|^\alpha}- \frac{b}{|p|^\beta} \right]=\sideset{}{'}\sum_{(m,n)\in \Z^2} \left[\frac{a}{Q_L(m,n)^{\frac{\alpha}{2}}}- \frac{b}{Q_L(m,n)^{\frac{\beta}{2}}} \right].
$$
\end{defi}
\begin{remark}
In the proof of Theorem \ref{mainthm} and Proposition \ref{prop:bound}, we will use a computer to check the values of a quotient $\mathcal{Q}_{\alpha,\beta}(L)$ among unit density lattices $L$. For that, we will obviously need to work with $\zeta_L^N(s)$ defined by \eqref{def:partialEpstein} for which a good approximation is known (see e.g. \cite{CrandallFastEval}). Since our goal is to work with exponents that are rather large (typically $(\alpha,\beta)=(12,6)$), only few terms (i.e. small $N$, which actually depends on $y$ in the parametrization in $\mathcal{D}$) are sufficient to accurately approximate $\zeta_L(s)$ by $\zeta_L^N(s)$.
\end{remark}

\subsection{Results on Epstein zeta functions and Lennard-Jones type energies}\label{sec:LJrecall}

In this section, we summarize important results concerning the Lennard-Jones type energies that we will need for our proof. The reader can refer to \cite{BetTheta15,Beterloc,OptinonCM,SamajTravenecLJ} for more details, proofs and results, including the optimality of the triangular lattice for $E_f$ for small exponents shown in \cite[Theorem 1.2.B.2]{BetTheta15}. We first start by a straightforward scaling formula.

\begin{lemma}[Scaling formulas]\label{lem:scale}
For any $L\in \mathcal{L}_2(1)$ and any $V>0$, $\sqrt{V}L \in \mathcal{L}_2(V)$ and, furthermore, for any $s>2$ and any $\alpha>\beta>2$, we have
\begin{equation}\label{eq:homog}
\zeta_{\sqrt{V} L}(s)=V^{-\frac{s}{2}}\zeta_L(s),\quad \textnormal{and}\quad E_f[\sqrt{V}L]= a V^{-\frac{\alpha}{2}}\zeta_L(\alpha)- b V^{-\frac{\beta}{2}}\zeta_L(\beta).
\end{equation}
\end{lemma}

We also recall the famous optimality result found by Rankin-Ennola-Cassels-Diananda in \cite{Rankin,Eno2,Cassels,Diananda} and generalized by Montgomery in \cite{Mont}.

\begin{thm}[Optimality of the triangular lattice for the Epstein zeta function]\label{thm:mont}
For any $s>2$, $\mathsf{A}_2$ is the unique minimizer of $L\mapsto \zeta_L(s)$ in $\mathcal{L}_2(1)$, up to rotation.
\end{thm}

Finally, we state the variational properties of $E_f$ regarding the minimizer among the dilated versions of a lattice $L$ as well as an upper bound for the co-volume of a global minimizer of $E_f$ (see \cite{BetTheta15,OptinonCM}).

\begin{thm}[Minimality properties of the Lennard-Jones type energy]\label{thm:LJ}
Let $\alpha>\beta>2$. For any $L\in \mathcal{L}_2(1)$, the unique minimizer of $V\mapsto E_f[\sqrt{V} L]$ is
\begin{equation}\label{eq:VL}
V_L:=\left( \frac{a \alpha \zeta_L(\alpha) }{b \beta \zeta_L(\beta)}\right)^{\frac{2}{\alpha-\beta}},
\end{equation}
and the minimal energies among the dilated versions of $L$ is therefore
$$
\min_{V>0}  E_f[\sqrt{V} L] =  E_f[\sqrt{V_L} L]=\frac{b^{\frac{\alpha}{\alpha-\beta}}\zeta_L(\beta)^{\frac{\alpha}{\alpha-\beta}}}{a^{\frac{\beta}{\alpha-\beta}}\zeta_L(\alpha)^{\frac{\beta}{\alpha-\beta}}} \left( \frac{\beta}{\alpha} \right)^{\frac{\beta}{\alpha-\beta}}\left( \frac{\beta}{\alpha}-1   \right)<0.
$$
In particular, the minimizer of $L\mapsto E_f[\sqrt{V_L} L]$ in $\mathcal{L}_2(1)$ only depends on $(\alpha,\beta)$.

\medskip

\noindent Furthermore, if $L_0\in \mathcal{L}_2(V_{L_0})$ is a global minimizer of $E_f$ in $\mathcal{L}_2$, then its volume $V_{L_0}$ is bounded as follows for any set of parameters $(\alpha,\beta,a,b)$:
\begin{equation}\label{eq:UBV}
0<V_{L_0}\leq \left( \frac{a \alpha}{b \beta}\right)^{\frac{2}{\alpha-\beta}}.
\end{equation}
\end{thm}

\begin{remark}
In the above theorem, notice that the minimizer of $L\mapsto E_f[\sqrt{V_L}L]$ (i.e. the ``shape" of the minimizer of $E_f$, see \cite[Sec. 1.1]{OptinonCM}) does not depend on $(a,b)\in (0,+\infty)$ but the minimal value of the energy $E_f$ does indeed depend on these parameters.
\end{remark}

\section{Proof of our results and generalization}\label{proof}

\subsection{Description of our general computer-assisted method}\label{sec:strategy}

We aim to present a proof of Theorem \ref{mainthm} and Proposition  \ref{prop:bound} as general as possible in order to be able to use our computer-assisted method for any pair of exponents $(\alpha,\beta)$. By Theorem \ref{thm:LJ}, we know that the ``shape" of the global minimizer of $E_f$, i.e. the minimizer of $L\mapsto E_f[\sqrt{V_L} L]$ in $\mathcal{L}_2(1)$ where $V_L$ is given by \eqref{eq:VL}, is independent of $(a,b)$. Therefore, we choose $a=1$ and $b=\frac{\alpha}{\beta}$ in such a way that the Lennard-Jones potential
$$
f(r)=\frac{1}{r^{\alpha}}-\frac{\alpha}{\beta}\frac{1}{r^\beta}
$$
reaches its minimum for $r=1$. Furthermore, it also follows that any global minimizer of $E_f$ has its volume smaller than $1$ according to \eqref{eq:UBV} since $\frac{a\alpha}{b\beta}=1$. Our goal is therefore to show that the triangular lattice $\sqrt{V}\mathsf{A}_2$ is the unique minimizer of $E_f$ in $\mathcal{L}_2(V)$ for all $V\leq 1$ and Theorem \ref{mainthm} will be proved. This is actually equivalent, by a simple scaling argument, with the following statement generalizing Proposition \ref{prop:bound}:
\begin{itemize}
\item  For any $(\alpha,\beta,a,b)\in (0,\infty)^4$, if $V\leq \left(\frac{a\alpha}{b\beta}  \right)^{\frac{2}{\alpha-\beta}}$ then the triangular lattice $\sqrt{V}\mathsf{A}_2$ is the unique minimizer of $E_f$ in $\mathcal{L}_2(V)$, up to rotation.
\end{itemize}
This automatically holds if the proof we are presenting in this section is applicable. In particular, we will show this result for $(\alpha,\beta)=(12,6)$ in Section \ref{sec:126} and for other exponents in Section \ref{sec:otherexponents}. To summarize, Theorem \ref{mainthm} is actually a corollary of Proposition \ref{prop:bound} as a simple consequence of Theorem \ref{thm:LJ}. However, we have chosen to present Theorem \ref{mainthm} as the main result of this paper.

\medskip

We remark that, using the homogeneity formulas stated in \eqref{eq:homog}, for $V>0$, we have
\begin{align*}
& \quad\quad\quad E_f[\sqrt{V} L]-E_f[\sqrt{V}  \mathsf{A}_2]\geq 0, \forall L\in \mathcal{L}_2(1)\\
&\iff V^{-\frac{\alpha}{2}}\left( \zeta_L(\alpha)-\zeta_{\mathsf{A}_2}(\alpha)\right)\geq \frac{\alpha}{\beta}V^{-\frac{\beta}{2}}\left( \zeta_L(\beta)-\zeta_{\mathsf{A}_2}(\beta)\right), \forall L\in \mathcal{L}_2(1) \\
& \iff V\leq V_{\alpha,\beta}:=\left( \frac{\beta}{\alpha}\right)^{\frac{2}{\alpha-\beta}} \inf_{L\in \mathcal{L}_2(1) \atop L\neq \mathsf{A}_2} \left( \frac{\zeta_L(\alpha)-\zeta_{\mathsf{A}_2}(\alpha)}{\zeta_L(\beta)-\zeta_{\mathsf{A}_2}(\beta)} \right)^{\frac{2}{\alpha-\beta}},
\end{align*}
where the inequality in the second equivalence follows from Theorem \ref{thm:mont}, i.e. the optimality of $\mathsf{A}_2$ in $\mathcal{L}_2(1)$ for $L\mapsto \zeta_L(s)$ for all $s>2$. If we show that $V_{\alpha,\beta}\geq 1$, then the theorem is proved. We therefore call, for any $L\in \mathcal{L}_2(1)\backslash \{\mathsf{A}_2 \}$,
\begin{equation}
\displaystyle \mathcal{Q}_{\alpha,\beta}(L):=\frac{\zeta_L(\alpha)-\zeta_{\mathsf{A}_2}(\alpha)}{\zeta_L(\beta)-\zeta_{\mathsf{A}_2}(\beta)},
\end{equation}
and $V_{\alpha,\beta}\geq 1$ is equivalent with
$$
\mathcal{Q}_{\alpha,\beta}(L)>\frac{\alpha}{\beta}, \quad \forall L\in  \mathcal{L}_2(1)\backslash \{\mathsf{A}_2 \}.
$$
More precisely, using the parametrization in the half-fundamental domain (see Proposition \ref{prop:D}), we naturally restrict our study to $L$ parametrized by $(x,y)\in \mathcal{D}$ defined by \eqref{def:D}. We therefore write $\mathcal{Q}_{\alpha,\beta}(x,y)$ instead of $\mathcal{Q}_{\alpha,\beta}(L)$ when $L$ has such parametrization and we aim to show, in three steps presented below, that
\begin{equation}\label{eq:Qxyineq}
\mathcal{Q}_{\alpha,\beta}(x,y)>\frac{\alpha}{\beta},\quad \forall (x,y)\in \mathcal{D}\backslash \left\{ \left(\frac{1}{2},\frac{\sqrt{3}}{2}\right)\right\}.
\end{equation}

\noindent \textbf{Step 1. Smoothness at $(x,y)=(1/2,\sqrt{3}/2)$.} We first show that $\mathcal{Q}_{\alpha,\beta}\in C^1(\mathcal{D})$. The smoothness on $\mathcal{D}\backslash \{(1/2,\sqrt{3}/2) \}$ is clear since $L\mapsto \zeta_L(s)$ is a smooth function on $\mathcal{D}$ when $s>2$. The continuity at $(1/2,\sqrt{3}/2)$ follows from the fact that $\mathcal{Q}_{\alpha,\beta}(x,y)$ converges to a positive constant as $(x,y)\to (1/2,\sqrt{3}/2)$ when $(x,y)\in \mathcal{D}\backslash \{(1/2,\sqrt{3}/2) \}$. Indeed, it has been proven by Coulangeon and Sch\"urmann in \cite{Coulangeon:2010uq} that $\mathsf{A}_2$ is a strict local minimum in $\mathcal{L}_2(1)$ of the Epstein zeta function $L\mapsto \zeta_L(s)$ for all $s>2$, which means by L'Hospital's Rule that
$$
\lim_{(x,y)\to (1/2,\sqrt{3}/2) \atop (x,y)\in \mathcal{D}\backslash \{(1/2,\sqrt{3}/2) \} } \mathcal{Q}_{\alpha,\beta}(x,y)=\lim_{L\to \mathsf{A}_2 \atop L\in \mathcal{L}_2(1)\backslash \{ \mathsf{A}_2\}} \frac{\zeta_L(\alpha)-\zeta_{\mathsf{A}_2}(\alpha)}{\zeta_L(\beta)-\zeta_{\mathsf{A}_2}(\beta)}>0,
$$
since both numerator and denominator are positive. Furthermore, for any $z\in \{x,y\}$ and any $L$ represented by $(x,y)\in \mathcal{D}$ , we have that
\begin{align}
\partial_z \mathcal{Q}_{\alpha,\beta} (x,y)=
\left(\frac{\zeta_L(\alpha)-\zeta_{\mathsf{A}_2}(\alpha) }{\zeta_L(\beta)-\zeta_{\mathsf{A}_2}(\beta)}\right)\left(\frac{\partial_z \zeta_L(\alpha)}{\zeta_L(\alpha)-\zeta_{\mathsf{A}_2}(\alpha)} - \frac{\partial_z \zeta_L(\beta)}{\zeta_L(\beta)-\zeta_{\mathsf{A}_2}(\beta)} \right) \label{eq:partialzQ}
\end{align}
As seen above, the first factor of this expression is going to a positive constant. By L'Hospital's Rule again, the second factor goes to the difference of quotients of the $z$-derivatives of the quadratic form associate to the 2nd order derivative of $L\mapsto \zeta_L(s)$ at $L=\mathsf{A}_2$, $s\in \{\alpha,\beta\}$, divided by the same quadratic form at the same point. These quotients are both finite by smoothness of $L\mapsto \zeta_L(s)$ and by the strict minimality of $\mathsf{A}_2$ for these Epstein zeta functions (see  \cite{CoulSchurm2018} and \cite[Eq. (6.6)]{OptinonCM} for an analogous computation involving Epstein zeta functions). We therefore obtain that
%
%
$$
\forall z\in\{x,y\},\quad \lim_{(x,y)\to (1/2,\sqrt{3}/2) \atop (x,y)\in \mathcal{D}\backslash \{(1/2,\sqrt{3}/2) \}} \partial_z \mathcal{Q}_{\alpha,\beta} (x,y) < \infty.
$$
We therefore obtain that $\nabla_{(x,y)} \mathcal{Q}_{\alpha,\beta}(x,y)$ is continuous on $\mathcal{D}$, which means that $\mathcal{Q}_{\alpha,\beta}\in C^1(\mathcal{D})$.

\medskip
\medskip

\noindent  \textbf{Step 2. Restriction to a compact set.} Since $\mathcal{D}$ is an infinite set (in the $y$-direction), we cannot check \eqref{eq:Qxyineq} numerically. However, since, for any $(x,y)\in \mathcal{D}$,
$$
\mathcal{Q}_{\alpha,\beta}(x,y)=\frac{\displaystyle y^{\frac{\alpha}{2}}\sideset{}{'} \sum_{m,n} \frac{1}{\left((m+xn)^2+ y^2 n^2 \right)^{\frac{\alpha}{2}}} - \zeta_{\mathsf{A}_2}(\alpha)}{\displaystyle y^{\frac{\beta}{2}} \sideset{}{'} \sum_{m,n} \frac{1}{\left((m+xn)^2+ y^2 n^2 \right)^{\frac{\beta}{2}}} - \zeta_{\mathsf{A}_2}(\beta)},
$$
it is straigthforward that $\lim_{y\to +\infty} \mathcal{Q}_{\alpha,\beta}(x,y)=+\infty$, which means that \eqref{eq:Qxyineq} is satisfied at least for $y$ large enough and any $x\in [0,1/2]$. More precisely, we have the explicit value $y_{\alpha,\beta}$ for such $y$ given by the following lemma whose proof is postponed to Section \ref{subsec:prooflemmas} to increase the readability of this section.

\begin{lemma}\label{lem1}
For any $\alpha>\beta>2$, we have that, for all $(x,y)\in \mathcal{D}$,
$$
y> y_{\alpha,\beta} \Rightarrow \mathcal{Q}_{\alpha,\beta}(x,y)>\frac{\alpha}{\beta},
$$
where
$$
y_{\alpha,\beta}:=
\left\{
\begin{array}{ll} 
\displaystyle \frac{9}{2^{1+\frac{2}{\beta}}\zeta(2\beta)^{\frac{2}{\beta}}}\left( \zeta_{\Z^2}(\beta) + \sqrt{\zeta_{\Z^2}(\beta)^2 - \frac{2}{3^{\beta}}\zeta(2\beta)\left( 2\zeta_{\mathsf{A}_2}(\beta)-\zeta_{\mathsf{A}_2}(2\beta) \right)} \right)^{\frac{2}{\beta}} & \mbox{ if $\alpha=2\beta$,}\\ 
\displaystyle\frac{3^{\frac{\alpha}{\alpha-\beta}}}{2^{1+\frac{2}{\alpha-\beta}}}\left( \frac{\alpha\zeta_{\Z^2}(\beta)}{\beta \zeta(\alpha)} \right)^{\frac{2}{\alpha-\beta}} & \mbox{otherwise,}
\end{array}
\right.
$$
and where $\zeta(s):=\displaystyle \sum_{m\in \N} \frac{1}{m^s}$ is the Riemann zeta function. 
\end{lemma}
\begin{remark}[Difference between the two values of $y_{\alpha,\beta}$]
The second expression giving $y_{\alpha,\beta}$ still holds for $\alpha=2\beta$, but is larger. Actually, for $(\alpha,\beta)=(12,6)$, the two above expressions are quite similar (the difference is of order $10^{-3}$), but it has to be noticed that the difference increases with $\beta$. That is why we have derived a slightly better formula for $\alpha=2\beta$, more precisely for saving some computational time in Step 3.
\end{remark}
\noindent We therefore write
$$
\mathcal{K}_{\alpha,\beta}:= \{ (x,y)\in \R^2 : x\in [0,1/2], y\in [\sqrt{3}/2, y_{\alpha,\beta}], x^2+y^2\geq 1\},
$$
in such a way that \eqref{eq:Qxyineq} holds in $\mathcal{D}\backslash \mathcal{K}_{\alpha,\beta}$.

\medskip
\medskip

\noindent  \textbf{Step 3. Numerical checking of \eqref{eq:Qxyineq} in $\mathcal{K}_{\alpha,\beta}$.} It remains to show that \eqref{eq:Qxyineq} holds in $\mathcal{K}_{\alpha,\beta}$. For that, since we are only summing a finite number of terms $\zeta_L^N(s)$ (see \eqref{def:partialEpstein}) of the Epstein zeta functions where $s\in \{\alpha,\beta\}$, we need to know how to control the value of the rest. The quantities summed in the Epstein zeta functions we are considering are converging fast to zero and the error terms $\zeta_L(s)-\zeta_L^N(s)$ have been studied for instance in \cite{CrandallFastEval} where an algorithm is written to compute $\zeta_L(s)$ with any degree of accuracy. It is also clear that the larger $y$ is, the more terms we need (i.e. more $N$ should be large) since there are more points at small distance from each others in $L$. The same can be done for the Riemann zeta function $\zeta(\alpha)$, for which it is well-known that the error term is bounded by $\frac{\alpha}{\alpha-1}N^{1-\alpha}$ (see e.g. \cite[p. 3]{RiemannZeta}). Since $y_{\alpha,\beta}$ is written in terms of Epstein and Riemann zeta functions, it is therefore easy to compute a new approximating value $\overline{y_{\alpha,\beta}}$ such that, for chosen $k\in \N$,
$$
y_{\alpha,\beta}\leq \overline{y_{\alpha,\beta}} \quad \textnormal{and}\quad 10^k\overline{y_{\alpha,\beta}}\in \N.
$$
This upper bound $\overline{y_{\alpha,\beta}}$, and more precisely the integer $k$, can be changed according to the accuracy we use in the sequel. We now write 
\begin{equation}\label{eq:Kbar}
\overline{\mathcal{K}_{\alpha,\beta}}=\{ (x,y)\in \R^2 : x\in [0,1/2], y\in [\sqrt{3}/2, \overline{y_{\alpha,\beta}}], x^2+y^2\geq 1\}
\end{equation}
and we want to show that \eqref{eq:Qxyineq} holds in $\overline{\mathcal{K}_{\alpha,\beta}}\supset \mathcal{K}_{\alpha,\beta}$.

\medskip

Our method relies on the fact that one can only compute a finite number of values $\{\mathcal{Q}_{\alpha,\beta}(x_i,y_j)\}_{i,j}$ where $(x_i,y_j)\in \overline{\mathcal{K}_{\alpha,\beta}}$ in order to approximate $\min_{(x,y)\in \overline{\mathcal{K}_{\alpha,\beta}}} \mathcal{Q}_{\alpha,\beta}(x,y)$ with a sufficient degree of accuracy. Therefore, the fact that $\min_{i,j} \mathcal{Q}_{\alpha,\beta}(x_i,y_j)>\frac{\alpha}{\beta}$ will imply \eqref{eq:Qxyineq}. We actually create a grid of $G_{\alpha,\beta}^\delta$ of $I\times J$ points $(x_i,y_j)\in \overline{\mathcal{K}_{\alpha,\beta}}$, as illustrated in Figure \ref{fig:ex}, such that:
\begin{itemize}
\item $x_1=0<x_2<...<x_I=1/2$ and  $y_1\approx\frac{\sqrt{3}}{2}<y_2<...<y_J=\overline{y_{\alpha,\beta}}$,
\item $ x_{i+1}-x_i=y_{j+1}-y_j=\delta= \frac{1}{2I}=\frac{\overline{y_{\alpha,\beta}}-y_1}{J}$ for all $i\in \{1,...,I-1\}$, for all $j\in \{1,...,J-1\}$,
\end{itemize}
and where $y_1$ is an approximation of $\frac{\sqrt{3}}{2}$ by above with the accuracy given by $\delta$, i.e. $|y_1-\sqrt{3}/2|<\delta$. We also choose $N$ such that $\zeta_L^N(s)$ is well approximated where $L$ is parametrized by any $(x,\overline{y_{\alpha,\beta}})$, $x\in [0,1/2]$, in such a way that any other lattice parametrized by $(x,y)\in \overline{\mathcal{K}_{\alpha,\beta}}$ is also well approximated.

\begin{figure}[!h]
\centering
\includegraphics[width=45mm]{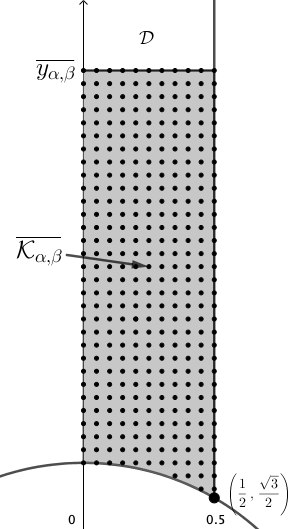}
\caption{Illustration of the fundamental domain $\mathcal{D}$ given by \eqref{def:D}, the compact set $\overline{\mathcal{K}_{\alpha,\beta}}\subset\mathcal{D}$ (in grey), the threshold value $\overline{y_{\alpha,\beta}}$ approximating by above $y_{\alpha,\beta}$ given by Lemma \ref{lem1}, the point $\left(\frac{1}{2},\frac{\sqrt{3}}{2}  \right)$ representing the triangular lattice $\mathsf{A}_2$ as well as a possible square grid $G_{\alpha,\beta}^\delta=\{(x_i,y_j)\}_{i,j}$ with increment $\delta$ where we compute the values of $\mathcal{Q}_{\alpha,\beta}$.}
\label{fig:ex}
\end{figure} 

\medskip

Since $\mathcal{Q}_{\alpha,\beta}\in C^1(\overline{\mathcal{K}_{\alpha,\beta}})$, it follows that $\mathcal{Q}_{\alpha,\beta}$ is a Lipschitz function in $\overline{\mathcal{K}_{\alpha,\beta}}$, i.e. there exists a positive constant $M_{\alpha,\beta}$ such that, for all $(x,y),(\bar{x},\bar{y})\in \overline{\mathcal{K}_{\alpha,\beta}}$,
\begin{equation}
| \mathcal{Q}_{\alpha,\beta}(x,y)- \mathcal{Q}_{\alpha,\beta}(\bar{x},\bar{y})|\leq M_{\alpha,\beta} |(x,y)-(\bar{x},\bar{y})|,
\end{equation}
which implies in particular that, for all $(x,y)\in \overline{\mathcal{K}_{\alpha,\beta}}$ and all $(x_i,y_j)\in G_{\alpha,\beta}^\delta$,
$$
\mathcal{Q}_{\alpha,\beta}(x,y)\geq \mathcal{Q}_{\alpha,\beta}(x_i,y_j)-M_{\alpha,\beta} |(x,y)-(x_i,y_j)|.
$$
This means that the latter inequality holds for $(x,y)$ in any square centred at $(x_i,y_j)$ with sidelength $\delta$ parallel to the axis. In that case, we have $|(x,y)-(x_i,y_j)|<\frac{\sqrt{2}}{2}\delta$. We now choose $\delta$ such that $M_{\alpha,\beta} \delta\frac{\sqrt{2}}{2}$ is small enough in such a way that we can numerically check that
$$
\min_{(x_i,y_j)\in G_{\alpha,\beta}^\delta}  \mathcal{Q}_{\alpha,\beta}(x_i,y_j)-M_{\alpha,\beta}\delta \frac{\sqrt{2}}{2}>\frac{\alpha}{\beta}.
$$
Therefore, one can be certain that \eqref{eq:Qxyineq} holds in $\overline{\mathcal{K}_{\alpha,\beta}}$ and the proof is completed.


\medskip

\textit{Computation of $M_{\alpha,\beta}$.} We do not give here any explicit formula for such $M_{\alpha,\beta}$ which covers all the possible case, but the following computation yields to a simple estimate given below in \eqref{def:Malphabeta}. We begins by writing, for $s\in \{\alpha,\beta\}$,
$$
\Delta_s(x,y):=\zeta_L(s)-\zeta_{\mathsf{A}_2}(s).
$$
Therefore $\mathcal{Q}_{\alpha,\beta}(x,y)=\frac{\Delta_\alpha(x,y)}{\Delta_\beta(x,y)}$ and we have $\forall (x,y)\in \overline{\mathcal{K}_{\alpha,\beta}}$, using \eqref{eq:partialzQ},
\begin{align*}
\| \nabla_{(x,y)} \mathcal{Q}_{\alpha,\beta}\|
\leq \mathcal{Q}_{\alpha,\beta}(x,y)\left\|\frac{\nabla_{L} \zeta_L(\alpha)}{\Delta_\alpha(x,y)}\right\| + \mathcal{Q}_{\alpha,\beta}(x,y)\left\|\frac{\nabla_{L} \zeta_L(\beta)}{\Delta_\beta(x,y)}\right\|.
\end{align*}
Moreover, we recall that $\Delta_s(x,y)\geq 0$ vanishes only for $(x,y)=(1/2,\sqrt{3}/2)$ (by uniqueness of the triangular minimizer) and $\nabla_L \zeta_L(s)$ vanishes only for $(x,y)\in \{(0,1), (1/2,\sqrt{3}/2)\}$ (the only critical points of the Epstein zeta function are the square and the triangular lattices).

\medskip

We can now bound above $\mathcal{Q}_{\alpha,\beta}(x,y)$ and $\left\|\frac{\nabla_{L} \zeta_L(s)}{\Delta_s(x,y)}\right\|$ for $(x,y)\in \overline{\mathcal{K}_{\alpha,\beta}}$  in the same way by splitting $\overline{\mathcal{K}_{\alpha,\beta}}$ into two parts:
\begin{itemize}
\item \textit{in a ball $B_\varepsilon(\mathsf{A}_2)$ (for the Euclidean norm) centred in $\mathsf{A}_2$ and with radius $\varepsilon$ (that we will fix afterwards)}. A Taylor expansion gives us here a simple upper bound  for $\left\|\frac{\nabla_{L} \zeta_L(s)}{\Delta_s(x,y)}\right\|$ in terms of the third (for the numerator) and the second derivative (for the denominator) of $L\mapsto \zeta_L(s)$ that can be easily bounded in the ball by iteration of the formula given below for $\| \nabla_{L} \zeta_L(s)\|$. The same Taylor expansion approach is used to bound above $\mathcal{Q}_{\alpha,\beta}$.
\item \textit{outside $B_\varepsilon(\mathsf{A}_2)$.} We roughly estimate the quotients by using the fact that, for $s>2$, we have, using the fact that $x\in [0,1/2]$ and $y\leq \overline{y_{\alpha,\beta}}$,
\begin{align*}
\| \nabla_{L} \zeta_L(s)\|^2&=\left(\partial_x \zeta_L(s)\right)^2 +\left( \partial_y\zeta_L(s)\right)^2\\
&=\left(\partial_x \left\{ y^{\frac{s}{2}}\sum_{m,n} \frac{1}{\left((m+xn)^2+y^2n^2 \right)^{\frac{s}{2}}} \right\} \right)^2 +\left(\partial_y \left\{y^{\frac{s}{2}}\sum_{m,n} \frac{1}{\left((m+xn)^2+y^2n^2 \right)^{\frac{s}{2}}} \right\} \right)^2\\
&=\left( -\frac{s}{2} y^{\frac{s}{2}} \sum_{m,n}\frac{2n(m+xn)}{\left((m+xn)^2+y^2n^2 \right)^{\frac{s}{2}+1}} \right)^2+ \left(\frac{s}{2}y^{\frac{s}{2}-1}\sum_{m,n} \frac{(m+xn)^2-y^2n^2}{\left((m+xn)^2+y^2n^2 \right)^{\frac{s}{2}+1}} \right)^2\\
&=\frac{s^2}{4}y^{s-2}\left\{ y^2\left(\sum_{m,n}\frac{2n(m+xn)}{\left((m+xn)^2+y^2n^2 \right)^{\frac{s}{2}+1}}\right)^2 + \left( \sum_{m,n} \frac{(m+xn)^2-y^2n^2}{\left((m+xn)^2+y^2n^2 \right)^{\frac{s}{2}+1}}\right)^2  \right\}\\
&\leq 2^s s^2\overline{y_{\alpha,\beta}}^{s-2}\left\{ \overline{y_{\alpha,\beta}}^2\left(\sum_{m,n}\frac{2|n|(|m|+0.5|n|)}{\left(m^2+n^2\right)^{\frac{s}{2}+1}}\right)^2 + \left( \sum_{m,n} \frac{m^2+0.25 n^2+|m||n|}{\left(m^2+n^2 \right)^{\frac{s}{2}+1}}\right)^2  \right\}\\
&=:S_{\alpha,\beta}(s),
\end{align*}
and using a simple lower/upper bounds for $\Delta_s(x,y)$ -- depending on $\varepsilon$ for the lower bound. In particular, we easily have, by growth of $L\mapsto \zeta_L(s)$ in the $y$-direction and its degrowth in the $x$-direction (see \cite{Mont}), for any $(x,y)\in\overline{\mathcal{K}_{\alpha,\beta}}$,
$$
\Delta_\alpha(x,y)=\zeta_L(\alpha)-\zeta_{\mathsf{A}_2}(\alpha)\leq \zeta_{\overline{L}}(\alpha),
$$
where $\overline{L}$ is parametrized by the point $(0,\overline{y_{\alpha,\beta}})\in \mathcal{D}$.
\end{itemize}
Finally, choosing $\varepsilon$ such that $\Delta_s(x,y)>1$ -- which is straightforward from the variation of $L\mapsto \zeta_L(s)$ (decreasing in the $x$-direction and increasing in the $y$-direction as shown in \cite{Mont}) --  allows, using estimates on the second derivative of $L\mapsto \zeta_L(s)$, to find a constant $M_{\alpha,\beta}$ that we approximate in order to make it numerically tractable. We observe that this constant is given by the maximum between the values we found in the neighborhood of $\mathsf{A}_2$ and basically the values in the upper boundary of $\overline{\mathcal{K}_{\alpha,\beta}}$. For the examples we cover in this article, we found that
\begin{equation}\label{def:Malphabeta}
M_{\alpha,\beta}=\zeta_{\overline{L}}(\alpha) \left(S_{\alpha,\beta}(\alpha)+ S_{\alpha,\beta}(\beta)\right).
\end{equation}

%

\begin{remark}[\textbf{List of the needed approximations - A summary}]
To summarize, we actually need to estimate the following quantities in order to complete our Step 3:
\begin{itemize}
\item \textbf{the Epstein zeta functions $\zeta_L(\alpha), \zeta_L(\beta)$},  that are approximated by $\zeta_L^N(\alpha),\zeta_L^N(\alpha)$ given in \eqref{def:partialEpstein}, and for which a bound of the error terms is known (see e.g. \cite{CrandallFastEval}). Notice that, since we are considering only sufficiently large exponents that are not treated in \cite{BetTheta15} (for which we already know that Theorem \ref{mainthm} holds), all our sums are converging very fast to zero.
\item \textbf{the value of $\overline{y_{\alpha,\beta}}$} based on the real $y_{\alpha,\beta}$ given by Lemma \ref{lem1} which is written in terms of
\begin{itemize}
\item the Epstein zeta function (see previous point),
\item the Riemann zeta function $\zeta(\alpha)$ for which the error term is bounded by $\frac{\alpha}{\alpha-1}N^{1-\alpha}$ is well-known (see e.g. \cite[p. 3]{RiemannZeta}).
\end{itemize}
\item \textbf{the upper bound $M_{\alpha,\beta}$} of $|\nabla_{(x,y)} \mathcal{Q}_{\alpha,\beta}(x,y)|$ in the compact set $\overline{\mathcal{K}_{\alpha,\beta}}$ -- for which we need a very rough estimate as given by \eqref{def:Malphabeta} -- that depends on
\begin{itemize}
\item the values of $S_{\alpha,\beta}(s)$ for $s\in \{\alpha,\beta\}$ that we estimate by its first terms when $|m|\leq N$ and $|n|\leq N$, as we do for $\zeta_{\overline{L}}(\alpha)$ by $\zeta^N_{\overline{L}}(\alpha)$.
\item the value of $\overline{y_{\alpha,\beta}}$ (see previous point).
\item in the following computations, we will round the approximation to the upper integer part. It does not really matter since we therefore choose $\delta$ with respect to $M_{\alpha,\beta}$ and the computation time stays reasonable.
\end{itemize}
\end{itemize}
Therefore, the fact that all the expressions we are summing or integrating are going to zero very fast allows us to choose some reasonable numbers $N$ (in the sense of ``not too large") and $\delta$ (in the sense of ``not too small") for computing our quantities and be sure that the inequality $\mathcal{Q}_{\alpha,\beta}(x,y)>\frac{\alpha}{\beta}$ holds in $\overline{\mathcal{K}_{\alpha,\beta}}$ whereas it is established in Lemma \ref{lem1} that the inequality holds in $\mathcal{D}\backslash \overline{\mathcal{K}_{\alpha,\beta}}$.
\end{remark}

\begin{remark}[\textbf{Possible improvements of our method}]
We obviously do not claim that our estimates are the best one. The bound $y_{\alpha,\beta}$ given in Lemma \ref{lem1} might be improved and a precise value of $M_{\alpha,\beta}$ could be computed by using more optimal bounds on $Q_L(m,n)$. However, in dimension $d=2$, we can easily compute a large number of values such that the total computational time is less than an hour. This is already quite long, but rather far from the 19 hours needed in Sarnak-Strombergsson's method \cite{SarStromb}. Of course, a three-dimensional adaption (if any) of such method should lead to better estimates in order to considerably reduce this computational time.
\end{remark}

\begin{remark}[\textbf{Non-adaptability to other potentials and in other dimensions}]\label{rmk:nonadapt}
It has to be noticed that our method strongly relies on the homogeneity of the Epstein zeta function (Lemma \ref{lem:scale}) as well as the fact that $f$ is a one-well potential which implies the bound \eqref{eq:UBV}. It is therefore unclear whether our method can be adapted to other interactions such as Morse type potentials. Furthermore, the optimality of the triangular lattice for the Epstein zeta function given in Theorem \ref{thm:mont} is also a key point in our method. Such optimality result is only available in dimensions $d\in \{2,8,24\}$ (see \cite{CKMRV2Theta}) but nothing has been shown in dimension $d=3$ so far (see e.g. \cite{SarStromb,Beterminlocal3d} or the numerical study in \cite{LBSTRieszLJ21}). This makes our method impossible to apply in dimension $d=3$ yet, whereas computational time is generally extremely long for such computer-assisted method in dimensions 8 and 24.
\end{remark}

\subsection{Application to the $(12,6)$ Lennard-Jones potential}\label{sec:126}

We now check these three steps for $(\alpha,\beta)=(12,6)$ in order to show Proposition \ref{prop:bound} which implies Theorem \ref{mainthm}. We use the software Scilab to perform our computations. Choosing $N=40$, we therefore find 
$$
y_{12,6}\leq \overline{y_{12,6}}=7.52,
$$
and $\zeta_L(6), \zeta_L(12)$ are accurate with an error of order $10^{-4}$ for $L$ parametrized by $(x,\overline{y_{12,6}})$. Using \eqref{def:Malphabeta}, we compute our Lipschitz constant (approximated by its upper integer part)
$$
M_{12,6}=181,
$$ and it is actually enough to compute our values with an increment of $\delta=10^{-2}$ (in that case we have $M_{12,6}\delta \sqrt{2}/2\approx 1.28$). Therefore, we split $[0,1/2]$ into $I=50$ equidistant values and $\left[ 0.87, 7.52\right]$, where $y_1=0.87\approx \frac{\sqrt{3}}{2}$ (with an accuracy of $\delta$), into $J=666$ equidistant values (both with increment $\delta$). We numerically checked that
$$
 \min_{(x_i,y_j)\in G_{12,6}^\delta}  \mathcal{Q}_{12,6}(x_i,y_j) - M_{12,6}\delta \frac{\sqrt{2}}{2}>\frac{12}{6}=2
$$
which completes the proof of Proposition \ref{prop:bound} and Theorem \ref{mainthm}.

\subsection{Application to other exponents and justification of Conjecture \ref{conj}}\label{sec:otherexponents}

The global optimality of a triangular lattice for $E_f$ has been conjectured to hold for any pair of exponents $(\alpha,\beta)$ (see \cite{BetTheta15,Beterloc,OptinonCM,SamajTravenecLJ}) and we have used the above strategy for a lot of them, which supports this conjecture. In particular, we have checked that our proof works for 
\begin{equation}\label{eq:exp}
(\alpha,\beta)\in  \{(k,6) : k\in \{14,16,18,20,22,24\} \},
\end{equation}
i.e. when we keep the Van der Waals attraction $-br^{-6}$ whereas the repulsion parameters $\alpha$ varies. The range of exponents $\alpha$ is actually inspired by Kaplan's remark in \cite[p. 184]{Kaplan} about applicability of Lennard-Jones type potentials in physics and the fact that we kept the same value for $\beta$ is therefore only motivated by Physical arguments. The reader can also easily check the applicability of our method for her/his favorite pair of exponents $(\alpha,\beta)$.

\medskip

In order to complete Step 3 for the exponents given in \eqref{eq:exp}, as in the $(12,6)$ classical case we have chosen to keep $N=40, \delta=10^{-2}, I=50, J=666$ since $\alpha\geq 6$. All our numerical findings are summarized in Table \ref{table1}. Thus, for these exponents $(\alpha,\beta)$ and all $(a,b)\in (0,\infty)^2$, the triangular lattice $\sqrt{V_{\mathsf{A}_2}}\mathsf{A}_2$ where
$$
V_{\mathsf{A}_2}=\left( \frac{a \alpha \zeta_{\mathsf{A}_2}(\alpha) }{b \beta \zeta_{\mathsf{A}_2}(\beta)}\right)^{\frac{2}{\alpha-\beta}}
$$
is the unique minimizer of $E_f$ in $\mathcal{L}_2$, up to rotation, i.e. the analogue of Theorem \ref{mainthm} is shown for these parameters. Furthermore, for the same exponents and for all $(a,b)\in (0,\infty)^2$, $\sqrt{V}\mathsf{A}_2$ is the unique minimizer of $E_f$ in $\mathcal{L}_2(V)$, up to rotation, if
$$
0<V<\left(\frac{a\alpha}{b\beta}  \right)^{\frac{2}{\alpha-\beta}},
$$
which shows Proposition \ref{prop:bound} for the exponents given in \eqref{eq:exp}.

\begin{table}
\centering
\begin{tabular}[c]{|c| c| c| c| c|}
\hline $\alpha$ & $\beta$ & $\overline{y_{\alpha,\beta}}$  &  
$M_{\alpha,\beta}$\\
\hline $14$ & 6 & 5.23   &  95 \\
\hline $16$ & 6 &  4.18  &   72\\
\hline $18$ & 6 & 3.60    &   50\\
\hline $20$ & 6 &  3.22   & 39\\
\hline $22$ & 6 &  2.97   &   34\\
\hline$24$ & 6 & 2.77    &  33\\
\hline
\end{tabular}
\caption{For the exponents $(\alpha,\beta)$ given in \eqref{eq:exp}: values of $\overline{y_{\alpha,\beta}}$ and $M_{\alpha,\beta}$ (for $N=40$ in the approximations).}
\label{table1}
\end{table}

\medskip

Let us now explain why Conjecture \ref{conj} should hold. Indeed, showing that $\mathcal{Q}_{\alpha,\beta}(L)>\frac{\alpha}{\beta}$ for all $2<\beta<\alpha$ and all $L\in \mathcal{L}_2(1)$ is equivalent with showing that
$$
s\mapsto \frac{\zeta_L(s)-\zeta_{\mathsf{A}_2}(s)}{s}
$$
is strictly increasing on $(2,+\infty)$ for all $L\in \mathcal{L}_2(1)$, i.e.
$$
\partial_s \left(\frac{\zeta_L(s)-\zeta_{\mathsf{A}_2}(s)}{s}\right)>0, \quad \forall L\in \mathcal{L}_2(1), \quad \forall s>2.
$$
Computing the above derivative, we find
\begin{align*}
&\partial_s \left(\frac{\zeta_L(s)-\zeta_{\mathsf{A}_2}(s)}{s}\right)=\frac{s(\partial_s \zeta_L(s)-\partial_s \zeta_{\mathsf{A}_2}(s))-\zeta_L(s)-\zeta_{\mathsf{A}_2}(s)}{s^2}>0 \quad \forall L\in \mathcal{L}_2(1), \quad \forall s>2\\
&\iff s\partial_s\zeta_L(s)-\zeta_L(s)>s\partial_s\zeta_{\mathsf{A}_2}(s)-\zeta_{\mathsf{A}_2}(s)\quad \forall L\in \mathcal{L}_2(1), \quad \forall s>2.
\end{align*}
Since 
$$
s\partial_s\zeta_L(s)-\zeta_L(s)=s\sideset{}{'}\sum_{p\in L}\partial_s\left( \frac{1}{|p|^s}\right)-\sideset{}{'}\sum_{p\in L} \frac{1}{|p|^s}=-\sideset{}{'}\sum_{p\in L} \frac{s\log |p| + 1}{|p|^s},
$$
Conjecture \ref{conj} implies that $\mathcal{Q}_{\alpha,\beta}(L)>\frac{\alpha}{\beta}$ for all $2<\beta<\alpha$ and all $L\in \mathcal{L}_2(1)$ and therefore implies the optimality of the triangular lattice $\sqrt{V_{\mathsf{A}_2}}\mathsf{A}_2$ for $E_f$ for all $2<\beta<\alpha$.

\section{Appendix: Proof of Lemma \ref{lem1}}\label{subsec:prooflemmas}


Let us show how we obtain our threshold value $y_{\alpha,\beta}$.

\begin{proof}[Proof of Lemma \ref{lem1}]
We start the proof by showing that, for all $(m,n)\in \Z$ and all $(x,y)\in \mathcal{D}$,
$$
\frac{m^2+n^2}{2}\leq (m+xn)^2 + n^2 y^2\leq \frac{3}{2}m^2 + \left(\frac{3}{4}+y^2 \right)n^2,
$$
where we have combined Young's inequality with the fact that $x^2+y^2\geq 1$ and $x\in [0,1/2]$. It follows that, for all $L\in \mathcal{D}$,
\begin{align*}
\mathcal{Q}_{\alpha,\beta}(L)=\frac{\displaystyle y^{\frac{\alpha}{2}}\sideset{}{'} \sum_{m,n} \frac{1}{\left((m+xn)^2+ y^2 n^2 \right)^{\frac{\alpha}{2}}} - \zeta_{\mathsf{A}_2}(\alpha)}{\displaystyle y^{\frac{\beta}{2}}\sideset{}{'} \sum_{m,n} \frac{1}{\left((m+xn)^2+ y^2 n^2 \right)^{\frac{\beta}{2}}} - \zeta_{\mathsf{A}_2}(\beta)}&\geq \frac{\displaystyle y^{\frac{\alpha}{2}}\sideset{}{'} \sum_{m,n} \frac{1}{\left(\frac{3}{2} m^2 + \left(\frac{3}{4}+y^2 \right)n^2 \right)^{\frac{\alpha}{2}}} - \zeta_{\mathsf{A}_2}(\alpha)}{\displaystyle y^{\frac{\beta}{2}}\sideset{}{'} \sum_{m,n} \frac{2^{\frac{\beta}{2}}}{\left(m^2 + n^2 \right)^{\frac{\beta}{2}}} - \zeta_{\mathsf{A}_2}(\beta)}\\
&\geq\frac{\displaystyle \frac{2^{1+\frac{\alpha}{2}}y^{\frac{\alpha}{2}}}{3^{\frac{\alpha}{2}}}\zeta(\alpha) - \zeta_{\mathsf{A}_2}(\alpha)}{\displaystyle 2^{\frac{\beta}{2}} y^{\frac{\beta}{2}}\zeta_{\Z^2}(\beta) - \zeta_{\mathsf{A}_2}(\beta)},
\end{align*}
 where $2^{\frac{\beta}{2}} y^{\frac{\beta}{2}}\zeta_{\Z^2}(\beta) - \zeta_{\mathsf{A}_2}(\beta)>0$ is ensured by the fact that $\zeta_{\Z^2}(\beta)>\zeta_{\mathsf{A}_2}(\beta)$  for all $\beta>2$ by Theorem \ref{thm:mont}, and $(2y)^{\frac{\beta}{2}}\geq (\sqrt{3})^{\frac{\beta}{2}}>1$ since $y\geq \sqrt{3}/2$. We therefore have
$$
\frac{\displaystyle \frac{2^{1+\frac{\alpha}{2}}y^{\frac{\alpha}{2}}}{3^{\frac{\alpha}{2}}}\zeta(\alpha) - \zeta_{\mathsf{A}_2}(\alpha)}{\displaystyle 2^{\frac{\beta}{2}} y^{\frac{\beta}{2}}\zeta_{\Z^2}(\beta) - \zeta_{\mathsf{A}_2}(\beta)}>\frac{\alpha}{\beta}\iff \eta_1 y^{\frac{\alpha}{2}}-\eta_2 y^{\frac{\beta}{2}}+\eta_3>0,
$$
where $\{\eta_1,\eta_2,\eta_3\}$ are defined by
\begin{equation*}
\eta_1:= \frac{2^{1+\frac{\alpha}{2}}}{3^{\frac{\alpha}{2}}}\zeta(\alpha)>0,\quad \eta_2:=2^{\frac{\beta}{2}}\frac{\alpha}{\beta} \zeta_{\Z^2}(\beta)>0 \quad \textnormal{and}\quad \eta_3:=\frac{\alpha}{\beta}\zeta_{\mathsf{A}_2}(\beta)-\zeta_{\mathsf{A}_2}(\alpha).
\end{equation*}
Let us now show that $\alpha_3>0$ by showing that the function $f$ defined by
$$
f(s):=\frac{\zeta_{\mathsf{A}_2}(s)}{s}
$$
is decreasing on $(2,\infty)$. This will indeed imply that $\eta_3=\alpha(f(\beta)-f(\alpha))>0$ for all $\alpha>\beta>2$. Differentiating $f$ gives
$$
f'(s)=-\frac{\zeta_{\mathsf{A}_2}(s)}{s^2}-\frac{1}{s}\sideset{}{'} \sum_{p\in \mathsf{A}_2}\frac{\log |p|}{|p|^s}
$$
and noticing that $\log |p|\geq \log \sqrt{2/\sqrt{3}}>0$ for all $p\in \mathsf{A}_2\backslash \{0\}$ implies that $f'(s)<0$ for all $s>2$ and $f$ is therefore decreasing. In the general case, we remark that
$$
\eta_1 y^{\frac{\alpha}{2}}-\eta_2 y^{\frac{\beta}{2}}+\eta_3 >\eta_1 y^{\frac{\alpha}{2}}-\eta_2 y^{\frac{\beta}{2}}
$$
and $\eta_1 y^{\frac{\alpha}{2}}-\eta_2 y^{\frac{\beta}{2}}>0$ if and only if $y> y_{\alpha,\beta}:=\left( \frac{\eta_2}{\eta_1} \right)^{\frac{2}{\alpha-\beta}}$. Therefore, the case $\alpha\neq 2\beta$ is proved. Furthermore, if $\alpha=2\beta$, we can easily solve the inequality
$$
\eta_1 y^{\frac{\alpha}{2}}-\eta_2 y^{\frac{\beta}{2}}+\eta_3=\eta_1 X^2 - \eta_2 X +\eta_3>0,\quad \textnormal{where}\quad X=y^{\frac{\beta}{2}}.
$$
The discriminant of the polynomial $P(X)=\eta_1 X^2 - \eta_2 X + \eta_3$ is $\Delta=\eta_2^2-4\eta_1\eta_3>0$ and therefore
$$
X> \frac{\eta_2 + \sqrt{\eta_2^2-4\eta_1\eta_3}}{2\eta_1} \Rightarrow P(X)>0,
$$
which means that 
$$
y>y_{\alpha,\beta}:= \left(  \frac{\eta_2 + \sqrt{\eta_2^2-4\eta_1\eta_3}}{2\eta_1}\right)^{\frac{2}{\beta}} \Rightarrow \mathcal{Q}_{\alpha,\beta}(L)>\frac{\alpha}{\beta},
$$
and proof is completed.
\end{proof}


\noindent \textbf{Acknowledgement:} I would like to thank the Austrian Science Fund (FWF) and the German Research Foundation (DFG) for their financial support through the joint project FR 4083/3-1/I 4354 during my stay in Vienna. I am also grateful to the anonymous referees for their helpful suggestions.

{\small \bibliographystyle{plain}
\bibliography{Biblio}}
\end{document}